\documentclass [12pt]{article}

\usepackage{placeins}
\usepackage{amsmath,amsthm,amscd,amssymb}
\usepackage[linesnumbered]{algorithm2e}
\usepackage{bm}
\usepackage[small, centerlast, it]{caption}
\usepackage{epsfig}
\usepackage[titles]{tocloft}
\usepackage[utf8]{inputenc} 
\usepackage[T1]{fontenc}
\usepackage{verbatim} 
\usepackage[active]{srcltx} 
\usepackage{fancyhdr}
\usepackage{caption}

\bmdefine{\bz}{z}
\bmdefine{\bTheta}{\Theta}

\def\sp{\hskip -5pt}

\def\b1{{1\!\!1}}

\def\sH{{\mathsf H}}

\def\bC{{\mathbb C}}           

\def\bI{{\mathbb I}}

\def\bR{{\mathbb R}}

\def\beq{\begin{eqnarray}}
\def\eeq{\end{eqnarray}}

\newcommand{\bra}[1]{\langle{#1}|}
\newcommand{\ket}[1]{|{#1}\rangle}

\usepackage{sectsty}
\sectionfont{\normalsize}
\subsectionfont{\normalsize\normalfont\itshape}
\newtheoremstyle{thm}
{12pt}
{12pt}
{\itshape}
{}
{\itshape\bfseries}
{}
{1em}
{}

\theoremstyle{thm}
\newtheorem{theorem}{Theorem}
\newtheorem{lemma}[]{Lemma}
\newtheorem{proposition}[]{Proposition}
\newtheorem{definition}[]{Definition}
\newtheorem{ex}[]{Example}

\begin{document}                                             
   
\thispagestyle{empty}


\begin{center}
 {\Large{\bf {Learning adiabatic quantum algorithms for solving optimization problems}} }

\bigskip
\bigskip

{\sc {Davide Pastorello}$^a$ and {Enrico Blanzieri$^b$}}

\bigskip

{
$^a$ Department of Mathematics, University of Trento\\
Trento Institute for Fundamental Physics and Applications \\via Sommarive 14, 38123 Povo (Trento), Italy\\ ~~E-mail: d.pastorello@unitn.it \\[10pt]
$^b$ Department of Engineering and Computer Science, University of Trento, \\via Sommarive 14, 38123 Povo (Trento), Italy\\ ~~E-mail: enrico.blanzieri@unitn.it\\[10pt]
}
                                                       
\end{center}

\sp

\abstract{\noindent An adiabatic quantum algorithm is essentially given by three elements: An initial Hamiltonian with known ground state, a problem Hamiltonian whose ground state corresponds to the solution of the given problem and an evolution schedule such that the adiabatic condition is satisfied. A correct choice of these elements is crucial for an efficient adiabatic quantum computation.  In this paper we propose a hybrid quantum-classical algorithm to solve optimization problems with an adiabatic machine assuming restrictions on the class of available problem Hamiltonians. The scheme is based on repeated calls to the quantum machine into a classical iterative structure. In particular we present a technique to learn the encoding of a given optimization problem into a problem Hamiltonian and we prove the convergence of the algorithm. Moreover the output of the proposed algorithm can be used to learn efficient adiabatic algorithms from examples.}

\section{Introduction}

Adiabatic Quantum Computing (AQC) was initially proposed as an application of the adiabatic theorem to solve optimization problems \cite{Farhi, Farhi2}, then it was demonstrated that it is a universal model of quantum computing \cite{aharonov}. It offers promising implementation perspectives and interesting connections with condensed matter physics \cite{albash}. A computation is implemented encoding the solution of a given problem into the ground state of the \emph{problem Hamiltonian} $H_P$ and considering the time evolution of a quantum system described by the time-dependent Hamiltonian:
\beq
H(t)=[1-s(t)]H_I(t)+s(t)H_P\qquad t\in[0,\tau],
\eeq      
where $s:[0,\tau]\rightarrow[0,1]$ is a monotone smooth function such that $s(0)=0$ and $s(\tau)=1$, it is called \emph{evolution schedule} and $\tau$ is the total evolution time. $H_I$ is the \emph{initial Hamiltonian}, with a known ground state, which does not commute with $H_P$. Starting from the ground state of $H_I$ and providing that the evolution is sufficiently slow, the final state of the system is the ground state of $H_P$ with high probability. Thus an \emph{adiabatic algorithm} is given by the triple $(H_I, s, H_P)$.
At the end of the evolution a measurement process gives the output of the computation. 

For a successful computation the evolution time must be large enough to ensure the adiabatic condition without destroying the computation efficiency. In general the run time of AQC depends on the choices of $H_P$, $H_I$, and $s$, its rigorous determination is typically hard \cite{albash}. Assuming that the initial Hamiltonian and the evolution schedule are fixed, a crucial issue is encoding a given problem into a problem Hamiltonian that can be implemented with the available resources. Moreover, for a given problem different problem Hamiltonians exist so another issue is selecting the best one to decrease the evolution time. 

In this paper we propose a hybrid quantum-classical algorithm, called  \emph{Adiabatic Quantum Computing Learning Search} (AQCLS), to find the solution of a given optimization problem using an adiabatic quantum machine. However we do not encode the problem into a Hamiltonian a priori, instead AQCLS is based on a classical iterative structure with repeated runs of an adiabatic quantum machine and a learning mechanism to induce modifications of the problem Hamiltonian towards a better encoding. Assuming that not every self-adjoint operator can be selected as problem Hamiltonian but there are limitations imposed by the physical architecture, the algorithm implements a random search in the space of \emph{available problem Hamiltonians} and a {tabu-inspired search} in the space of problem solutions. By encoding energetic penalties of already-visited solutions within a simulated annealing structure, the search is guided towards the solution of the optimization problem and converges to a corresponding problem Hamiltonian. 

AQCLS shares with our recent \emph{Quantum Annealing Learning Search} (AQLS) \cite{ED} the iterative approach to define the problem representation into the quantum architecture and the tabu-inspired search strategy. The main difference derives from the different nature of the considered quantum architectures, an adiabatic computer is a universal machine \cite{aharonov} instead the quantum annealer considered in \cite{ED} is a specific-purpose machine. In fact the optimization problems tackled by AQCLS are not limited to \emph{Quadratic Unconstrained Binary Optimization} (QUBO) problems like done in AQLS. 

In the next section we briefly review the adiabatic theorem and some crucial issues in AQC that we try to encompass with the proposed technique. In Section \ref{AQLS} we introduce the main idea of the considered tabu-inspired search with the learning mechanism of problem encoding into a $H_P$. Section \ref{convergence} is devoted to the convergence proof.

\section{Hard tasks in AQC}

Let us briefly review the content of adiabatic theorem: Assume to prepare a quantum system in the ground state $\psi_0$  of a given Hamiltonian, then one can change the Hamiltonian smoothly in time. If the change is sufficiently slow then the system remains in the istantaneous ground state with high probability. The time-dependent Hamiltonian is described by a smooth one-parameter family of self-adjoint operators $\{H(t)\}_{t\in[0,\tau]}$ in the Hilbert space $\sH$ of the considered quantum system. The dynamics of the system that is prepared in the ground state $\psi_0\in\sH$ (i.e. the pure state given by the eigenvector corresponding to the minimum of the spectrum) is given by the solution of the Schr\"odinger equation
\beq\label{SE}
i\hbar\frac{d}{dt} \psi(t)=H(t)\psi(t)\qquad t\in[0,\tau],
\eeq
with the initial condition $\psi(0)=\psi_0$. Let us re-parametrize the time-dependent Hamiltonian as $\widetilde H(s):=H(\tau s)$ with the variable $s\in[0,1]$. For any value of $s$ we have the following eigenvalue problem:
\beq
\widetilde H(s)\ket{l,s}=E_l(s)\ket{l,s},
\eeq
where the index $l\in\{0,1,2,...\}$ labels the eigenvectors $\ket{l,s}$ and the eigenvalues $E_l(s)$ of $\widetilde H(s)$, so $E_0(s)$ is the minimum of the spectrum of $\widetilde H(s)$ and $\ket{0,s}$ the corresponding eigenvector, i.e. the ground state. 

A first formulation of the adiabatic theorem, assuming the non-degeneracy of the ground state for any $s\in[0,1]$, is the following:
\begin{theorem}
If $E_1(s)-E_0(s)>0$ for any $s\in[0,1]$ then:
\beq
\lim_{\tau\rightarrow +\infty} |\langle 0,1| \psi(\tau)\rangle|=1,
\eeq
where $\psi(\tau)$ is the solution of (\ref{SE}), with initial condition $\psi(0)=\ket{0,0}$, calculated in $t=\tau$.
\end{theorem}

\noindent
This statement ensures that the state of the evolving system remains close to the ground state of $H(t)$, for any $t\in[0,\tau]$, if the evolution time is large. 
In order to evaluate how big $\tau$ should be to obtain an acceptable probability to be in the ground state at the end of the evolution, we need a more refined result like the following (cf. \cite{teufel} and \cite{jensen}):

\begin{theorem}
Let $\{H(t)\}_{t\in[0,\tau]}$ be a time-dependent Hamiltonian and $\{\widetilde H(s)\}_{s\in[0,1]}$ be the re-scaled Hamiltonian such that $\lambda(s):=E_1(s)-E_0(s)>0$ for any $s\in[0,1]$.
\\
If:
\beq\label{adiabatic condition}
\tau\geq\frac{4}{\epsilon}\left[\frac{\parallel\dot{\widetilde H}(0)\parallel}{\lambda(0)^2}+\frac{\parallel\dot{\widetilde H}(1)\parallel}{\lambda(1)^2}+\int_0^1 ds \left(10\frac{\parallel\dot{\widetilde H}(s)\parallel}{\lambda^3}+\frac{\parallel\dot{\widetilde H}(s)\parallel}{\lambda}\right)\right]
\eeq
where $\lambda:=\min_s\lambda(s)$, $\epsilon\in(0,1)$ and $\parallel\,\,\,\parallel$ is the standard operator norm, then:
\beq
\parallel \psi(\tau)-\ket{0,1}\parallel\leq\epsilon,
\eeq
where $\psi(\tau)$ is the solution of (\ref{SE}), with initial condition $\psi(0)=\ket{0,0}$, calculated in $t=\tau$.
\end{theorem}

\noindent
This statement implies that for any finite evolution time we have a nonzero probability of ending in a state that is different from the ground state of the problem Hamiltonian. The existence of this \emph{quantum step} is important for the convergence of the algorithm presented here. Inequality (\ref{adiabatic condition}) gives an estimation of the run time of an adiabatic computation, it depends on the time derivative of the Hamiltonian and on the spectral gap $\lambda$. In particular this latter quantity is hard to compute in general \cite{mc}. For this reason, determining the run time is a hard task in AQC so that for the most known adiabatic algorithms the run times are not exactly determined \cite{albash}. 

Given an optimization problem whose solution minimizes a cost function $f\!:X\rightarrow \bR$, where $X$ is a finite set, a standard adiabatic algorithm that solves the problem is physically based on a quantum system described in the Hilbert space $\sH\simeq \bC^{|X|}$. The elements of a fixed basis of $\sH$, called \emph{computational basis}, represent all the possible solutions. The computation is realized by the adiabatic time evolution of the system according to the Hamiltonian:
\beq\label{Standard}
H(t)=\left(1-\frac{t}{\tau}\right)H_I+\frac{t}{\tau}H_P\qquad t\in[0,\tau],
\eeq
providing the system is prepared in the known ground state of $H_I$ at $t=0$. In (\ref{Standard}) there is the routine choice of a linear interpolation where $\tau$ is the total evolution time. The crucial issue for a correct initialization of the adiabatic quantum machine is encoding the given problem into $H_P$ such that the ground state corresponds to the optimum. The standard choice in AQC is selecting a problem Hamiltonian that is diagonal in the computational basis, however this choice may be too severe, e.g. for some quantum systems it may produce many-body localization with deleterious effects on algorithm efficiency \cite{mbl}. A convenient choice to improve the performance of AQC can be given by a non-diagonal $H_P$ \cite{non-diagonal}. 
In general, the class of \emph{available problem Hamiltonians} is given by those that can be directly realized in laboratory and Hamiltonians that can be efficiently simulated by quantum circuits \cite{berry}. 

A remarkable aspect of AQC is the choice of the evolution schedule. Even if the standard choice is the linear interpolation used in (\ref{Standard}), in some cases this could be a poor strategy in terms of the run time. The most known case is that of \emph{adiabatic search} in an unstructured database: If the solution is represented by the element $\ket x$ of the computational basis, an adiabatic algorithm can be defined by an initial Hamiltonian whose ground state is the equal superposition of all the basis states and by a problem Hamiltonian of the form $H_P=\bI-\ket x\bra x$. It is well-known that if the evolution schedule is $s(t)=\frac{t}{\tau}$ on the interval $[0,\tau]$ then the run time of the algorithm grows linearly in the dimension of the database as for a classical exhaustive search. On the other hand if $s$ has the behavior of a hyperbolic tangent, the algorithm presents the Grover's quantum speed-up \cite{adiabatic Grover}.  

The hard calculation of spectral gaps and the arbitrariness of choosing $H_I$, $H_P$ and $s$ prevent the actual existence of general guiding principles to improve the performance of adiabatic quantum algorithms and various attempts in this direction have been recently proposed \cite{non-diagonal, shortcut, gap amplification}. In the next section we propose a strategy that combines a nonstandard variation of tabu search and simulated annealing to learn efficient problem Hamiltonians for adiabatic quantum computations.

\section{Adiabatic Quantum Computing Learning Search}\label{AQLS}

In this section we describe the general scheme of a heuristic search for the solution of an optimization problem by means of repeated calls of an adiabatic quantum machine within a classical iterative structure.
   
The general optimization problem that we face is the minimization of an objective function $f:X\rightarrow\bR$ where $X$ is a finite set. We consider a quantum system described in the $|X|$-dimensional Hilbert space $\sH$ and let $\{\ket x\}_{x\in X}$ be the computational basis. Assume for now that the initial Hamiltonian $H_I$ for AQC is fixed and its known ground state $\ket{\Phi_0}$ is given by a coherent superposition of all the basis states:
\beq\label{phi0}
\ket{\Phi_0}:=\sum_{x\in X} b_x \ket x\quad\mbox{with}\quad b_x\not =0 \,\,\forall x\in X \,\,\,\mbox{and}\,\,\,  \sum_{x\in X} |b_x|^2=1.
\eeq
\noindent
Let us assume that an evolution schedule $[0,\tau]\ni t\mapsto s(t;\tau)\in[0,1]$ is fixed, where $\tau$ is the total evolution time. For instance $t\mapsto s(t;\tau)$ could be linear $s(t;\tau)=t/\tau$, or a \emph{smoothstep function}  like $s(t;\tau)=3(t/\tau)^2-2(t/\tau)^3$, anyway the value of $\tau$ determines how slow the Hamiltonian interpolation is.

The  
class of available problem Hamiltonians is described by a parametrized family of self-adjoint operators $\{H_P^{(w)}\}_{w\in W}$ where $W\subseteq \bR^n$ and $[H_P^{(w)}, H_I]\not=0$ for all $w\in W$. Thus we are assuming that the selection of a problem Hamiltonian is specified by $n$ real parameters, moreover we require that the function $w\mapsto H_P^{(w)}$ is continuous\footnote{Continuity w.r.t. Euclidean topology on $W$ and topology induced by the operator norm.} to ensure that a small perturbation of the control parameters corresponds to a small change of the problem Hamiltonian.

\begin{ex}
In order to give an example of a class of available Hamiltonians in a real machine, let us consider the well-known D-Wave machine.
However let us specify that the D-Wave machine is a \emph{quantum annealer}. In quantum annealing (QA) the quantum system is generally coupled to the environment, so the evolution is characterized by dissipation and decoherence, instead in this paper we consider adiabatic evolutions of a closed quantum system. Moreover QA does not assume that the entire computation takes place in ground state like in AQC \cite{mc}. Here we consider the D-Wave available Hamiltonians just to give an example of $\{H_P^{(w)}\}_{w\in W}$.
The D-Wave hardware architecture is given by a network of qubits arranged on the vertices of a 
sparse graph $(V,E)$ where the edges represent the couplings among them.
The available problem Hamiltonians are operators on the Hilbert space $\sH\simeq (\bC^{2})^{\otimes |V|}$ of the following form:
\beq
H_P^{(\theta_i, \theta_{ij})}=\sum_{i\in V} \theta_i\sigma_z^{(i)}+\sum_{(i,j)\in E} \theta_{ij} \sigma_z^{(i)}\sigma_z^{(j)},
\eeq
where $\sigma_z^{(i)}$ acts on the ith qubit as the Pauli matrix $\sigma_z$ and on the other tensor factors as the identity. In this case the available problem Hamiltonians are labelled by the real parameters $\theta_i$, $\theta_{ij}$ for $i,j\in V$.
\end{ex}
 
The scheme we propose here contains the \emph{generation phase} of candidate solutions, the evaluation of the objective function and the \emph{acceptance phase}. The first phase is quantum, the other two are classical. The generation phase is realized calling the adiabatic machine initialized with an available problem Hamiltonian $H_P^{(w)}$, then the evolution is set with a certain evolution time $\tau$ ending in the ground state $H_P^{(w)}$ with some probability $p$. At the beginning of the search the problem Hamiltonian is randomly selected thus it does not encode the properties of the objective function, moreover the evolution time could be too short to ensure an adiabatic evolution, so the procedure starts with a randomized generation of candidate solutions. The main idea is to equip the search with a \emph{learning mechanism} in order to modify the problem Hamiltonian used to generate a candidate towards a better encodings of the problem.

The keystone of the proposed learning mechanism is a tabu-inspired search based on the following argument: To impose a penalty on the element $\hat x\in X$ to discourage 
its election as candidate solution, we implement a Hamiltonian such that $\ket{\hat x}$ is an excited state far from the ground state. A simple choice is $H=A\ket{\hat x}\bra{\hat x}$ with $A>0$. 

\begin{definition}
Let $\{\ket {x_1},...,\ket{x_r}\}$ be a collection of quantum states from the computational basis. Assume repetitions are allowed, i.e. it may be $x_i=x_j$ for some $i\not =j$. The corresponding \textbf{tabu Hamiltonian} is defined as:
\beq\label{TH}
H_{tabu}:=A\sum_{I=1}^r \ket{x_i}\bra{x_i}\qquad \mbox{with}\qquad A>0.
\eeq
\end{definition}  

\noindent
$H_{tabu}$ is a diagonal operator in the computational basis with eigenvalues given by nonnegative multiples of $A$. Without loss of generality, in the following we set $A$ to 1, then $H_{tabu}$ presents eigenvalues that are nonnegative integers. Let us stress that (\ref{TH}) is an obvious choice of a Hamiltonian which energetically penalizes a set of states, a generalized tabu Hamiltonian is any diagonal operator admitting the given states as excited ones. 
\\
The ideas presented above are the core of the hybrid quantum-classical algorithm AQCLS (Algorithm \ref{Schema}).  
AQCLS is based on a random search among the available problem Hamiltonians that is guided by a tabu-inspired search in the solution space (updating the tabu Hamiltonian) towards a better encoding of the given problem. The candidate solutions of the problem are generated running the adiabatic machine with modified problem Hamiltonians, the evolution time gradually increases during the search. 
 Let us illustrate the proposed algorithm.

 \begin{algorithm}[ht!]\label{Alg}
\footnotesize

\vspace{-1cm}

\KwData{Family of available problem Hamiltonians $\{H_P^{(w)}\}_{w\in W}$, initial Hamiltonian $H_I$, evolution schedule $s(t;\tau)$}
\KwIn{$f(x)$ to be minimized, minimum evolution time $t_{min}$, evolution time increasing step $\nu$, initial vector of parameters $w_0$, maximum variance $\sigma^2_{max}$, variance decreasing rate $\eta$, number $N$ of iterations with constant variance and constant evolution time, evolution probability $q$, termination parameter $N_{max}$, maximum number $i_{max}$ of iterations}
\KwResult{$H_P$ problem Hamiltonian, $\tau$ evolution time, $x^*$ minimum point of $f(x)$ } 
\SetKwProg{ffun}{function}{:}{}
$\tau\gets t_{min}$, $H_{tabu}\gets 0$\;
$\sigma^2\gets\sigma^2_{max}$, $w^*\gets w_0$\;
randomly initialize two Hamiltonians $H_P^{(w_1)}$ and $H_P^{(w_2)}$ from $\{H_P^{(w)}\}_{w\in W}$ \;
prepare the state $\ket{\Phi_0}$ \tcp*[h]{\it defined in formula (\ref{phi0})};\\
evolve twice according to Hamiltonians $[1-s(t;\tau)]H_I+s(t;\tau)\,H_P^{(w_i)}$ with $i=1,2$\;
measure the final states w.r.t. computational basis obtaining outcomes $x_1$ and $x_2$; \tcp*[h]{\it the final states correspond to ground states of $H_P^{(w_1)}$ and $H_P^{(w_2)}$ with probabilities $p_1$ and $p_2$ respectively}\\ 
evaluate  $f(x_1)$ and $f(x_2)$\;
\If{$f(x_1)\not =f(x_2)$}{use the best to initialize $x^*$\  and the Hamiltonian $H_P^{(w^*)}$\;
use the worst to initialize $x'$\;
initialize the tabu Hamiltonian: $H_{tabu}=\ket{x'}\bra{x'}$\;
}
$d\gets 0; e\gets 0; i\gets 0$\; 
\Repeat{ $i=i_{max}$ or $d+e\geq N_{max}$ } {
\If{$N$ divides $i$}
{$\sigma^2 \gets \sigma^2-\eta\,\sigma^2$\;
$\tau\gets \tau+\nu$\;}
initialize Hamiltonian $H_P^{(w)}$ sampling $w$ according to the normal distribution $\phi_{w^*,\sigma^2}$\;
prepare the state $\ket{\Phi_0}$\; 
with probability $1-q$ measure w.r.t. the computational basis otherwise evolve according to the Hamiltonian $[1-s(t;\tau)]\,H_I+s(t;\tau)\,(H_P^{(w)}+H_{tabu})$ and measure
 w.r.t. computational basis. Find  the outcome ${x'}$\;
  \eIf{\emph{$x' \not= x^*$}} {
  evaluate $f(x')$\;
   \eIf{\emph{$f(x')<f(x^*)$}}{$swap(x',x^*)$; $H_P^{(w^*)} \gets H_P^{(w)}$; \tcp*[h]{\it \emph{$x'$} is better}\\ 
use $x'$ to update the tabu Hamiltonian: $H_{tabu} \gets H_{tabu}+\ket{x'} \bra{x'}$;}
 { $d\gets d+1$\;
  with probability $[\sigma^2/\sigma^2_{max}]^{(f(x')-f(x^*))}$ $swap(x',x^*)$; $H_P^{(w^*)} \gets H_P^{(w)}$; \tcp*[h]{\it \emph{$x'$} is worse (suboptimal acceptance)} 
 }
    }
 {$e\gets e+1$\;}   
      $i\gets i+1$\;
  }
\Return $H_P\equiv H_P^{(w^*)}+H_{tabu}$, $\tau$, $x^*$\;

\vspace{0.5cm}

\caption{\it Hybrid quantum-classical algorithm to solve an optimization problem learning a corresponding problem Hamitonian of an adiabatic quantum machine.}
\label{Schema}
\end{algorithm}
\FloatBarrier

\noindent
\textbf{Candidate solution generation and testing}. Initially there are a random selection of two Hamiltonians (Algorithm \ref{Schema}, line 3) and the preparation of the quantum hardware in the ground state of $H_I$ (line 4), then two evolutions are set with evolution time $\tau=t_{min}$ (line 5) in order to generate two candidate solutions by measurements on the final states (line 6). The probabilities $p_1$ and $p_2$ are not explicitly determined and depend on: $t_{min}$, $H_I$, the problem Hamiltonians, and the evolution schedule. The candidates are tested: The best $x^*$ is used to initialize the new problem Hamiltonian and the worst $x'$ is used to update the tabu Hamiltonian (lines 8-12).

\noindent
The new problem Hamiltonian $H_P^{(w)}$ is initialized from $x^*$ sampling $w\in W$ according to a normal distribution centered in $w^*$ (where $H^{(w^*)}_P$ is the problem Hamiltonian used to generate the best candidate) with variance $\sigma^2$:
\beq
\phi_{w^*,\sigma^2}(w)=\frac{1}{\sqrt {2\pi\sigma^2}} \exp\left[-\frac{\parallel w-w^*\parallel^2}{2\sigma^2}\right].
\eeq

 \noindent
In the iterative part (lines 15-34) new candidate solutions are repeatedly generated and tested.
 The variance periodically decreases during the search from a maximum value $\sigma_{max}^2$ to zero (line 16) and contextually the evolution time of the adiabatic machine increases from $t_{min}$. 
Lines 19 and 20 contain the initialization of the current problem Hamiltonian $H_P^{(w)}$ and the preparation of the ground state $\ket{\Phi_0}$ of $H_I$. The effect of line 21 is that, with probability $q\gg0$, the new candidate solution is generated setting the time evolution in $[0,\tau]$ according to the Hamiltonian $H(t)=[1-s(t;\tau)]\,H_I+s(t;\tau)\,(H_P^{(w)}+H_{tabu})$, or, with probability $1-q$, no evolution is set and a measurement is performed immediately after the preparation of $\ket{\Phi_0}$. These probabilistic alternatives are important for the algorithm convergence, as explained in the next section.
\\
\\
\textbf{Candidate solution acceptance}. After the evaluation of the generated candidate solution $x'$ (line 23), there are the comparison with the best solution so far $x^*$ and the application of the acceptance rule (lines 24-30). If $x'$ is better than $x^*$, it is accepted as current solution with probability 1 and tabu Hamiltonian is updated (line 25 and 26). If $x'$ is not better than $x^*$ then it is accepted with probability $[\sigma^2/\sigma^2_{max}]^{(f(x')-f(x^*))}$ (line 29).


\vspace{0cm}

\noindent
Thus the acceptation probability of $x'$, given $x^*$ as the best candidate so far, is:
\beq\label{AP}
P(x'|x^*)=\left\{
\begin{array}{cc}
1\hspace{3.3cm} & f(x')<f(x^*)\\

[\sigma^2/\sigma^2_{max}]^{(f(x')-f(x^*))} & f(x')\geq f(x^*)
\end{array}
\right.
\eeq

\noindent
that can be compared with the acceptation probability of Simulated Annealing (SA) at given temperature $T$:
\beq\label{SA}
P_{SA}(x'|x^*;T)=\left\{
\begin{array}{cc}
1\hspace{2cm} & f(x')<f(x^*)\,\\

e^{-\frac{f(x^*)-f(x')}{T}} & f(x')\geq f(x^*).  
\end{array}
\right.  
\eeq

\noindent
Therefore we can interpret the role of $\sigma$ as the temperature parameter $T$ of a simulated annealing process in the following way:  $T=-\log^{-1}(\sigma^2/\sigma_{max}^2)$. Mapping the classical part of AQCLS into a SA process is crucial for the proof of convergence of the algorithm in Section \ref{convergence}.
\\
\\
\textbf{Termination}. There are two termination conditions (line 35): One involves the achievement of the maximum number of iterations, the other corresponds to the convergence to a solution of the optimization problem. The counter $e$ (line 32) counts the number of consecutive times that the current solution $x^*$ is generated; the counter $d$ (line 28) counts the number of times that the current solution and the new candidate solution differ and the current one is not worse. When $d+e$ reaches the maximum value, the algorithm returns $x^*$ as optimum of $f$, the Hamiltonian $H_P\equiv H_P^{(w^*)}+H_{tabu}$ as the problem Hamiltonian whose ground state is $\ket{x^*}$ with high probability, and the corresponding evolution time $\tau$ of the adiabatic algorithm $(H_I, s(t;\tau), H_P)$.
\\
\\
As mentioned above, the variance of the normal distribution used to explore the space of available Hamiltonians can be interpreted as the temperature of a SA process, so it decreases during the search (line 16). On the other hand the evolution time increases during the search (line 17) towards the evolution time of the adiabatic algorithm defined by the limit problem Hamiltonian (whose existence is proved in the next section).  
The counter $i$ (line 33) counts the number of iterations, the \textbf{if} statement (lines 15-18) provides that for $N$ iterations we have a fixed values of variance and evolution time. The motivation of this choice is twofold: Regarding the variance, we obtain a structure of temperature levels of SA that is important for convergence, in the meanwhile we prevent that the evolution time grows too fast. The constant increasing rate $\nu$ of the evolution time is an arbitrary choice and line 17 can be generalized to $\tau\gets h(\tau)$ where $h$ is a monotone increasing function so that we have large increments at the beginning of the search and smaller afterwords. 
\\
Let us stress that the algorithm does not only return the problem solution but also provides a problem Hamiltonian encoding the given problem and the total evolution time of the corresponding adiabatic computation. This means that for a class of optimization problems it is possible  to run several times the algorithm on different instances and learn the initialization of the adiabatic quantum machine from examples. In this way the learned mapping could be used to initialize other instances of the same class of problems.    
\\
The family of available problem Hamiltonians $\{H_P^{(w)}\}_{w\in W}$ can be assumed to coincide with the whole set of Hamiltonians that can be physically implemented in the quantum machine. However an interesting problem is how one could restrict to a subfamily and consequently constrain the search in order to increase the adiabatic machine efficiency for a given objective function. Its solution could add a useful learning bias for some class of optimization problems, and further research is needed to investigate this issue.

\section{Convergence of AQCLS}\label{convergence}

\noindent
The classical part of AQCLS algorithm is based on a simulated annealing procedure with modified generation probabilities depending on a temperature parameter represented by the variance $\sigma^2$. In order to prove the convergence we apply some results about the convergence of SA with non-constant generation matrix and about Markov chains \cite{tabu, af1, af2} that are summarized in \cite{ED}. We will show that the convergence of our hybrid algorithm is implied by some properties of the SA-like classical part that are closely related to the features of the quantum part.  

In the previous section we have described the candidate solution generation realized by Algorithm \ref{Schema}. Given a current solution $x_i\in X$ and a value of the variance $\sigma^2$, there is a certain probability of generating the new candidate solution $x_j\in X$. In this sense the generation phase can be described by a family of $|X|\times |X|$ stochastic matrices $\{A(\sigma^2)\}_{\sigma^2>0}$ that we call \emph{variance-dependent generation matrix}, where the matrix element $a_{ij}(\sigma^2)$ is defined as the probability for generating the candidate solution $x_j$ given the current solution $x_i$ and variance $\sigma^2$. We are not interested in the explicit computation of generation probabilities instead let us introduce the notion of \emph{neighborhood graph}. 
\begin{definition}
Let $A$ be a stochastic matrix. The \textbf{neighborhood graph} induced by $A$ is the directed graph $G_{A}=(X,E)$ where $E:=\{(x_i,x_j) \,: \, a_{ij}>0 \}$.
\\
$G_A$ is said to be:
\\
i) \textbf{complete} if any pair of vertices is connected by a bidirectional edge;
\\
ii) \textbf{strongly connected} if for any pair of vertices $(x_i, x_j)$ there is a directed path from $x_i$ to $x_j$.
\end{definition}

Obviously completeness implies strong connectivity. The neighborhood graph associated to the generation phase of Algorithm \ref{Schema} is complete, as proved below, with a crucial effect on the algorithm convergence.

\begin{lemma}\label{compl}
The neighborhood graph $G_{A(\sigma^2)}$ associated to the candidate solutions generation of Algorithm \ref{Schema} is complete for any $\sigma^2>0$.
\end{lemma}

\begin{proof}
The candidate solution generation calls the adiabatic quantum machine, and it does not rule out any solution because of the \emph{quantum step} implied by theorem \ref{adiabatic condition} and the probabilistic alternatives of line 21. Therefore $a_{ij}(\sigma^2)\not= 0$ for any $i,j\in\{1,2,...,|X|\}$ and any $\sigma^2>0$. 
\end{proof}

As anticipated in the algorithm description of the previous section, the calls to the adiabatic machine are embedded into a classical SA structure realized by an acceptance phase (lines 24-30). Let us recall that the variance $\sigma^2$ can be interpreted as a temperature parameter in these terms $T=-\log^{-1}(\sigma^2/\sigma_{max}^2)$, for this reason we apply some results about the convergence of SA processes characterized by a temperature-dependent generation matrix. 
In the practice of SA with temperature-dependent generation matrix, one can decrease temperature in the following way: He chooses a sequence of temperature levels $T_1\geq T_2\geq \dots \geq T_K\simeq 0$ running many iterations for each level as clarified in \cite{tabu}. Algorithm \ref{Schema} realizes the temperature levels in lines 15-18.

\begin{proposition}\label{prop1}
Algorithm \ref{Schema} converges.
\end{proposition}

\begin{proof}
Let us consider a classical SA algorithm with temperature-dependent generation matrix $\{A(T)\}_{T>0}$ and acceptance probability $P_{SA}$ as defined in (\ref{SA}).
The SA algorithm can be modeled by a stochastic process with temperature-dependent transition matrix $\{M(T)\}_{T>0}$ defined by $m_{ij}(T)=a_{ij}(T)P_{SA}(x_i|x_j; T)$, where the state space represents the solution space of the considered optimization problem.  
Now let us follow the approach of \cite{tabu} where general results about Markov chains are applied \cite{af1, af2} to study the convergence of SA processes: For a fixed value of $T>0$, $M(T)$ is the transition matrix of a Markov chain, if the neighborhood graph $G_{A(T)}$ induced by $A(T)$ is strongly connected for any $T$ then $M(T)$ has a unique \emph{stationary distribution}\footnote{Let $M$ be the transition matrix of a Markov chain with state space $X=\{x_1,...,x_n\}$. The \emph{stationary distribution} $\pi$ of $M$ is a probability distribution on $X$ satisfying:
$\pi(x_j)=\sum_{i=1}^n m_{ij}\pi(x_i)$ $ \forall j=1,...,n$.}
 $\pi_T$ for any $T$ and the limit $\pi^*=\lim_{T\rightarrow 0} \pi_T$ exists. Therefore, if the property of strong connectivity is satisfied the stochastic process describing the SA algorithm converges as $T\rightarrow 0$ with limit distribution $\pi^*$. Algorithm \ref{Schema} realizes the structure of a SA with a temperature-dependent (variance-dependent indeed) generation matrix $\{A(\sigma^2)\}_{\sigma^2>0}$, in view of Lemma \ref{compl} the associated neighborhood graph is complete then strongly connected for any $\sigma^2$. As a consequence Algorithm \ref{Schema} converges as $\sigma^2\rightarrow 0$.

\end{proof}

The result above only ensures that AQCLS admits a limit distribution on the solution space $X$, now the crucial point is proving that such distribution is nonzero only on the optimum solutions of the given problem. For this goal, let us consider a statement that summarizes some convergence results \cite{tabu} on generalized SA processes with temperature-dependent generation probabilities.

\begin{proposition}\label{SAconvergence}
Let $\{A(T)\}_{T>0}$ be a temperature-dependent generation matrix and $P_{SA}$  be the acceptance probability defined in (\ref{SA}) by the objective function $f$. 
\\
If the following hypotheses are satisfied:
\\
\\
1. $A(T)$ is combinatorially symmetric for any $T>0$, i.e.:
$$a_{ij}(T)>0 \,\,\Leftrightarrow \,\,a_{ji}(T)>0\qquad \forall i, j , T. $$
2. There exists a $\delta>0$ such that for each $T>0$:
$$a_{ij}(T)>0\,\,\Rightarrow\,\, a_{ij}(T)\geq \delta\quad\mbox{whenever}\quad i\not =j.$$ 
3. The neighborhood graph of $A(T)$ does not depend on $T$.
\\
\\
then $\pi^*(x)>0$ if and only if $x\in\emph{arg}\!\min f$.
\end{proposition}

\noindent
The statement above is a variation of Theorem 1.3 of \cite{tabu} where the hypothesis of \emph{weak reversibility} of the neighborhood graph for any $T$ is substituted by the hypothesis of combinatorial symmetry of $A(T)$ for any $T$. Combinatorial symmetry is a stronger condition than weak reversibility defined in \cite{tabu}.\\ 
Now we are in position to prove our main result:

\begin{proposition}\label{conv}
Algorithm \ref{Schema} converges to a global optimum of $f$.
\end{proposition}

\begin{proof}
Let us check that the $\sigma^2$-dependent generating probabilities of the SA-like process implemented by Algorithm \ref{Schema} satisfy the hypotheses of Proposition \ref{SAconvergence}. As provided by Lemma \ref{compl} the neighborhood  graph associated to candidate solutions generation is complete for any $\sigma^2$ then Hypotheses 1 and 3 are satisfied. In order to verify also Hypothesis 2, let us observe that the probability $p$ of generating the candidate $x'$ given any current solution $x^*\not = x'$ satisfies:
\beq
p\geq (1-q) \min_{x\in X} |b_x|^2,
\eeq
where $1-q$ is the probability to perform a measurement process immediately after the preparation of $\ket{\Phi_0}$ (line 21) and $b_x$ is the coefficient of the state $\ket x$ in the coherent superposition $\ket{\Phi_0}$. Therefore Hypothesis 2 of Proposition \ref{SAconvergence} is satisfied for $\delta=(1-q) \min_{x\in X} |b_x|^2$.
\end{proof}

Proposition \ref{conv} states that the AQCLS algorithm asymptotically returns a global minimum of the objective function $f$. Moreover the output provides a problem Hamiltonian $H_P$ whose ground state is $\ket{x^*}$ and an adiabatic evolution time $\tau$ to produce that result. Therefore one can run several times the algorithm with the same data in order to find more efficient $H_P$ and $\tau$.


\section{Conclusions}

In this work we have proposed a hybrid quantum-classical algorithm to solve an optimization problem implementing a heuristic search where an adiabatic quantum machine is repeatedly called to generate candidate solutions. The scheme is equipped with a tabu-inspired mechanism which guides the search towards a better encoding of the problem into a problem Hamiltonian. 

AQALS considers AQC and general optimization problems whereas QALS \cite{ED} deals with quantum annealing and QUBO problems. In addition to the broader class of objective functions, another remarkable difference between AQCLS and QALS is that in AQCLS the evolution time of the adiabatic machine increases during the search and its final value is an output, instead in QALS the annealing time is implicitly assumed to be constant. Note that AQCLS is not a generalization of QALS but has to be considered an analogous scheme for a different quantum architecture.
    
 We have proven the convergence of the AQCLS algorithm, so it returns a global optimum, a problem Hamiltonian, and a suitable evolution time that give the information for obtaining the solution as an output of a proper adiabatic algorithm. Therefore the algorithm can be run several times (with the same $H_I$ and $s(t,\tau)$ or with different ones) in order to learn adiabatic algorithms $(H_I, s(t,\tau), H_P)$ by examples to solve a certain class of optimization  problems.
Roughly speaking AQCLS is a hybrid algorithm which can be applied to learn adiabatic quantum algorithms for solving optimization problems.

\section*{Acknowledgements}
 
The present work is supported by:

\vspace{-0.4cm}

\begin{center}
 \includegraphics[width=4cm]{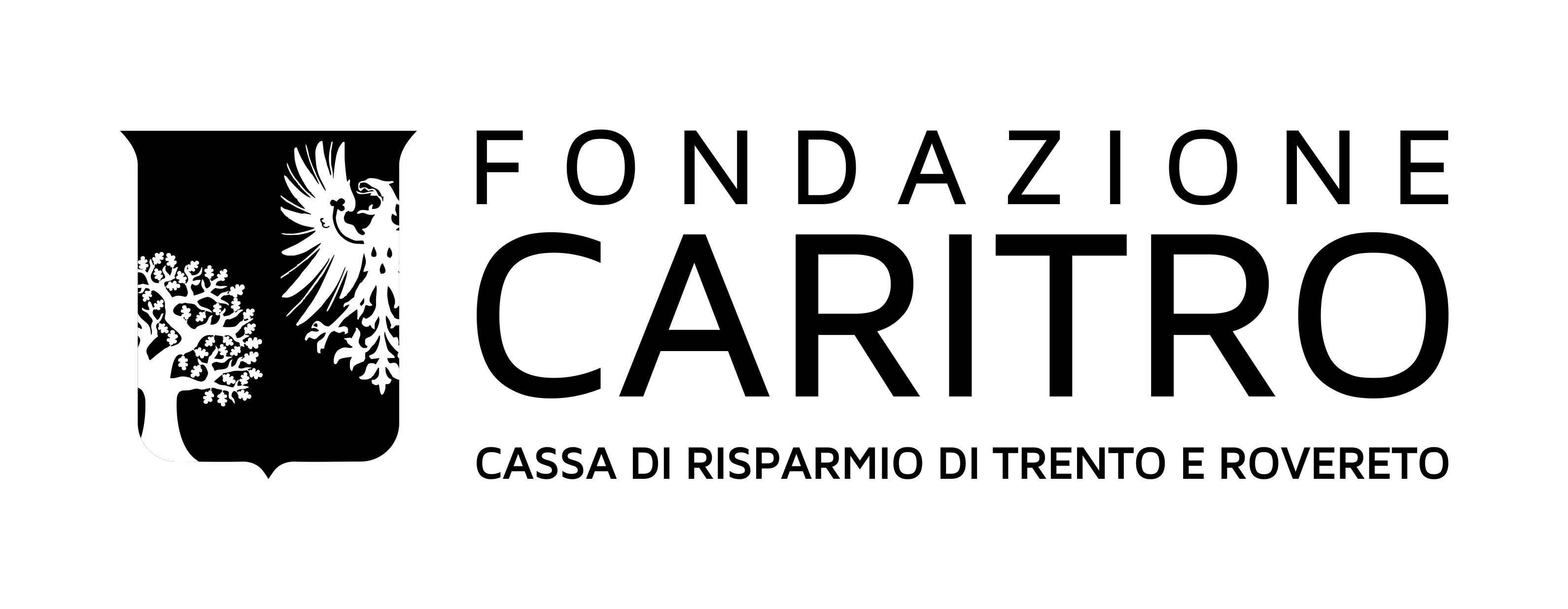}
\end{center}

\end{document}